\documentclass[a4paper,11pt]{article}
\usepackage[utf8x]{inputenc}

\usepackage{indentfirst}
\usepackage{amsmath}
\usepackage{amssymb}
\usepackage{amsfonts}
\usepackage{amsthm}
\usepackage{amsopn}
\usepackage[pdftex]{graphicx}
\usepackage{graphicx}
\usepackage{epsfig}
\usepackage{fancyhdr}
\usepackage{rotating}
\usepackage{overpic}
\usepackage{ucs}
\usepackage{verbatim} 
\usepackage{enumerate}
\newlength{\defbaselineskip}
\newcommand{\setlinespacing}[2]%
           {\setlength{\baselineskip}{#1 \defbaselineskip}}

\newtheorem{thm}{Theorem}[section]
\newtheorem{lem}[thm]{Lemma}

\newtheorem{definition}[thm]{Definition}
\newtheorem{fact}[thm]{Fact}

\newlength{\btw}
\newlength{\stw}
\title{Computing bounded-width tree and branch decompositions of $k$-outerplanar graphs}
\author{Ioannis Katsikarelis}
\date{}
\begin{document}

\maketitle

\begin{abstract}
 By a well known result \cite{bodlaender1988planar} the treewidth of $k$-outerplanar graphs is at most $3k-1$. This paper gives, besides a rigorous proof of this fact, an algorithmic implementation of the proof, i.e.\ it is shown that, given a $k$-outerplanar graph $G$, a tree decomposition of $G$ of width at most $3k-1$ can be found in $O(kn)$ time and space. Similarly, a branch decomposition of a $k$-outerplanar graph of width at most $2k+1$ can be also obtained in $O(kn)$ time, the algorithm for which is also analyzed.
\end{abstract}

\section{Introduction}
Tree decompositions are an important subject of algorithmic studies on graphs and the notion of treewidth has become one of the main parameters used in complexity characterizations of many problems related to them. Both notions were introduced by Robertson and Seymour \cite{Robertson198449}. A \emph{tree decomposition} of a graph is a structural representation of the graph using a tree and sets of vertices of the graph (sometimes called `bags'). These sets are associated with the nodes of the tree and all vertices of the graph are included in some sets. The size of the largest such set (minus one) is called the width of the tree decomposition. Informally, the \emph{treewidth} of a graph is a measure of how tree-like the graph is, based on some tree decomposition of it.

The main reason behind the interest in tree decompositions is that they enable ideas for solving problems that are tractable when the input graph is a tree, but intractable for the general case of input graphs (usually NP-hard), to be extended and applied in the general case as well. After a tree decomposition of small width is found for an input graph, some form of dynamic programming algorithm efficiently solves the problem on the tree decomposition, making use of its tree-like structure (see \cite{bodlaender1994tourist, bodlaender1997treewidth}). The problem of finding a tree decomposition of small width, however, was shown to be NP-hard \cite{Arnborg:1987:CFE:37170.37183} and in practice, dominates the running time of most algorithms based on tree decompositions.

Despite this intractability result, Bodlaender \cite{bodlaender1988planar} has shown that a tree decomposition of width at most $3k-1$ exists, if the input graph is $k$-outerplanar. Several authors have cited \cite{bodlaender1988planar} as also providing a linear time algorithm (of $O(kn)$ complexity) for computing a tree decomposition of such width, apart from the upper bound on the treewidth (some examples are included in \cite{Alber200426, Chen200120, Kammer:2007:DSK:1778580.1778615, KammerTholey2009}). In fact, due to its constructive nature, Bodlaender's proof (presented also in \cite{Bodlaender98apartial}) can be seen as implying the existence of such an algorithm, yet the computations involved do not allow for an obvious linear time performance, calling for a more detailed, in-depth investigation.

By a non-trivial application of the ideas presented in \cite{bodlaender1988planar, Bodlaender98apartial} a rigorous proof of the upper bound of at most $3k-1$ on the treewidth of $k$-outerplanar graphs is given here, as well as the presentation of a corresponding algorithm that finds in $O(kn)$ time and space a tree decomposition of a $k$-outerplanar graph of width at most $3k-1$. It should be noted also that the upper bound on the treewidth of $k$-outerplanar graphs has been shown to be tight, i.e.\ there exist $k$-outerplanar graphs with treewidth $3k-1$ \cite{KammerTholey2009}. 

\paragraph{} A closely related notion to the treewidth of a graph is that of its \emph{branchwidth}. It is based on a \emph{branch decomposition} of the graph, which bears similarities to -- and can also be used much akin to -- a tree decomposition. These two notions were also introduced by Robertson and Seymour \cite{Robertson1991153}. For some otherwise intractable problems, in the same way as with tree decompositions, first a branch decomposition of bounded width is found for a given graph and then some dynamic programming algorithm efficiently solves the problems based on the structure of the branch decomposition.

Instead of mapping groups of vertices of the input graph to all the nodes of a tree, a branch decomposition makes use of a bijective mapping between the edges of the graph and the leaf nodes of a tree. As is the case for treewidth, the problem of determining whether a given graph $G$ has branchwidth at most $b$ is NP-complete, when both $G$ and $b$ are considered parts of the input \cite{Seymour1994_rat}.

The ideas by Bodlaender \cite{bodlaender1988planar} mentioned above, can also be used for the design of an algorithm that computes a branch decomposition of a given $k$-outerplanar graph, of width at most $2k+1$, as introduced in \cite{hjortas2005}. This algorithm requires linear time and is also presented and analyzed here, mainly for completeness.

\paragraph{}This paper is structured as follows. Section 2 gives some definitions of notions used in the following sections, while Section 3 contains Bodlaender's proof of the upper bound on the treewidth of $k$-outerplanar graphs, which heavily influences the algorithms that follow. The algorithm that finds a bounded-width tree decomposition is presented in Section 4 and finally, Section 5 discusses matters relating to branchwidth.

\section{Definitions}
Definitions of concepts used in the following sections are given here, starting with the formal definitions of tree and branch decompositions.

\begin{definition}\label{def:treewidth}
 A \emph{tree decomposition} of a graph $G=(V,E)$ is a pair $(\{X_i|i\in I\},T=(I,F))$, with $\{X_i|i\in I\}$ a family of subsets of $V$, one for each node of $T$, and $T$ a tree such that:
 \begin{itemize}
  \item $\bigcup_{i\in I}X_i=V$.
  \item for all edges $(v,w)\in E$, there exists an $i\in I$ with $v\in X_i$ and $w\in X_i$.
  \item for all $i,j,k\in I$: if $j$ is on the path from $i$ to $k$ in $T$, then $X_i\cap X_k\subseteq X_j$.
 \end{itemize}
The \emph{width} of a tree decomposition $(\{X_i|i\in I\},T=(I,F))$ is $\max_{i\in I}|X_i|-1$. The \emph{treewidth} of a graph $G$ is the minimum width over all possible tree decompositions of $G$.
\end{definition}

\begin{definition}\label{def:branchwidth}
 A \emph{branch decomposition} of a graph $G=(V,E)$ is a pair $(T=(I,F),\sigma)$, where $T$ is a tree with every node in $T$ of degree one or three, and $\sigma$ is a bijection from $E$ to the set of leaf nodes in $T$. The \emph{order} of an edge $f\in F$ is the number of vertices $v\in V$, for which there exist adjacent edges $(v,w),(v,x)\in E$, such that the path in $T$ from $\sigma(v,w)$ to $\sigma(v,x)$ uses $f$. These vertices comprise what is called the \emph{middle set}, denoted by $\omega(f)$. The \emph{width} of branch decomposition $(T=(I,F),\sigma)$ is the maximum order over all edges $f\in F$. The \emph{branchwidth} of $G$ is the minimum width over all branch decompositions of $G$.
\end{definition}

The \emph{outerplanarity} index of a plane graph (an embedded planar graph) is defined as the maximum distance of any vertex of the graph to the outer face, by calling two vertices adjacent when they share a face, and the outer face adjacent to all vertices on the outer face. An alternative, equivalent definition of outerplanarity is the following \cite{Baker:1994:AAN:174644.174650}.

\begin{definition}\label{def:outerplanar}
 An embedding of a graph $G=(V,E)$ is \emph{$1$-outerplanar}, if it is planar and all vertices lie on the exterior face. For $k\ge2$, an embedding of a graph $G=(V,E)$ is \emph{$k$-outerplanar}, if it is planar and removal of all vertices on the outer face yields a $(k-1)$-outerplanar embedding of the resulting graph. A graph is \emph{$k$-outerplanar}, if it has a $k$-outerplanar embedding.
\end{definition}

The notion of a \emph{minor} is used in the proof of the next section.

\begin{definition}\label{def:minor}
 A graph $G=(V,E)$ is a \emph{minor} of a graph $H=(W,F)$, if $G$ can be obtained from $H$ by a series of vertex deletions, edge deletions, and edge contractions, where an edge contraction is the operation that replaces two adjacent vertices $v$, $w$ by one that is adjacent to all vertices that were adjacent to $v$ or $w$.
\end{definition}

Finally, the notions of \emph{vertex} and \emph{edge remember number}  of maximal spanning forests\footnote{Meaning a spanning tree of every connected component.} of a graph are also used. They were introduced in \cite{bodlaender1988planar}.

\begin{definition}\label{def:vr_er}
 Let $T=(V,F)$ be a maximal spanning forest of a graph $G=(V,E)$. A fundamental cycle is associated with every edge $e=(v,w)\in E-F$, meaning the unique cycle that consists of $e$ and the simple path from $v$ to $w$ in $T$. The \emph{vertex remember number} of $G$, relative to $T$, denoted by $vr(G,T)$, is defined as the maximum over all $v\in V$ of the number of fundamental cycles that use $v$. Similarly, the \emph{edge remember number} of $G$, relative to $T$, denoted by $er(G,T)$, is defined as the maximum over all edges $e\in E$ of the number of fundamental cycles that use $e$.
\end{definition}

\section{Treewidth of $k$-outerplanar graphs}
In this section the complete proof of the upper bound is presented, which will be used as justification for the correctness of the algorithm described in the next section.  After some brief notes regarding the results, a series of intermediary lemmas is provided building towards Theorem \ref{thm:thm83_tw_upbound}.

Lemmas \ref{thm:lem78_tw2} and \ref{thm:lem16_tw_minor} are well known results. The exposition of the proof follows the same lines as that of Bodlaender \cite{bodlaender1988planar,Bodlaender98apartial}, with the proof of Lemma \ref{thm:lem82_degree} expanded for complete accuracy. Lemma \ref{thm:lem78_tw2} deals with the case where $k=1$, while the results following it generalize to $k\ge 2$.

\begin{lem}\label{thm:lem78_tw2}
 Every outerplanar graph $G=(V,E)$ has treewidth at most $2$.
\end{lem}
\begin{proof}
Since all outerplanar graphs have degeneracy at most 2 \cite{Lick1970}, there must exist at least one vertex $v$ with degree either equal to 1, or equal  to 2. Supposing that $degree(v)=2$, let $w$ and $x$ be the two neighboring vertices of $v$, i.e.\ $(v,w)\in E$, $(v,x)\in E$, $w\not= x$. The graph $G'=(V-\{v\},(E-\{(v,w),(v,x)\})\cup \{(w,x)\})$, obtained by removing $v$ and connecting $w$ and $x$, is an outerplanar graph. Assuming, with induction, that there exists a tree decomposition $(\{X_i|i\in I\},T=(I,F))$ of $G'$ with treewidth $\le2$, it can be seen that there must be some $i\in I$, with $w\in X_i\wedge x\in X_i$. Now let $i^*\notin I, I^*=I\cup\{i^*\}, X_{i^*}=\{v,w,x\}$ and $T^*=(I^*,F\cup\{(i,i^*)\})$, i.e.\ a new node $i^*$ is added adjacent to $i$, with bag $\{v,w,x\}$. It's easily verified that $(\{X_i|i\in I^*\},T^*)$ is a tree decomposition of $G$ with treewidth at most $2$.
\end{proof}

\begin{lem}\label{thm:lem82_degree}
 For every $k$-outerplanar graph $G=(V,E)$, there exists a $k$-outerplanar graph $H=(V',E')$, such that $G$ is a minor of $H$, and degree$(H)\le3$.
\end{lem}
\begin{proof}
 To find the $k$-outerplanar graph $H$, one can replace every vertex of degree $d\ge4$ in $G$ by a path of $d-2$ vertices of degree $3$, in such a way that $H$ remains $k$-outerplanar and edge ordering in $H$ matches that in $G$. In the following, this transformation is called an \emph{expansion}. For the expansion of a vertex to maintain $k$-outerplanarity, the \emph{layers} of all faces of $G$ must stay the same in $H$. The layers of the faces of a planar graph are assigned as follows: First, set the layer of the outer face to 0. Then, for each face, set its layer to be 1 higher than the layer of its adjacent face with the lowest layer.
 
 An expansion does not maintain facial adjacency and thus, may lead to a face becoming `disconnected' from a crucial neighboring face, which might imply an increase in the outerplanarity index of $H$, relative to that of $G$. This can be easily avoided, however, if one notices two facts about expansion: First, before expansion is applied to a vertex, all related faces are adjacent to each other (sharing the vertex to be expanded) and so their layers can only be at most one level apart. Second, if the expansion is applied in such a way, that a face whose layer is of the smallest level is placed adjacent to all others, then face layering is maintained.
  
 It turns out it is always possible to place faces in this way, a fact illustrated by Figure \ref{fig:degree_construction}. Suppose that face $f_i$ is of lower, or equal layer to all others in $G$. Placing that face adjacent to all others in $H$, as seen in Figure \ref{fig:degree_construction}, and the rest in their corresponding order according to $f_i$, will ensure that face layering does not change, and so $H$ remains $k$-outerplanar.
 
 \begin{figure}[htbp]
 \centerline{\includegraphics[width=120mm]{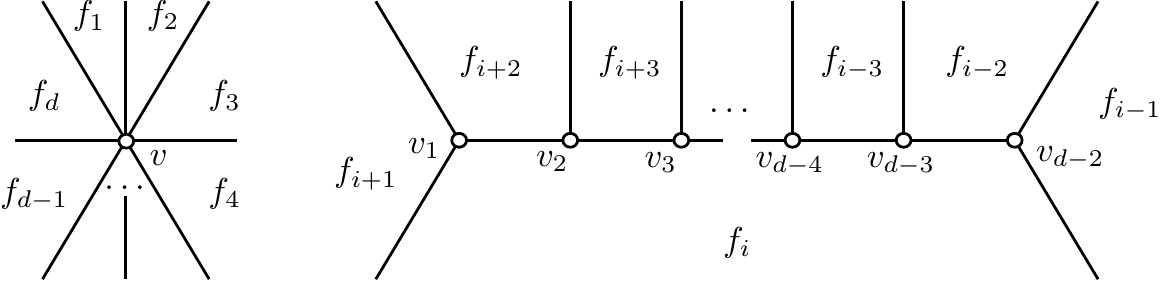}}
 \caption{The expansion described in the proof of Lemma \ref{thm:lem82_degree}. On the left is the vertex to be expanded, while on the right is the resulting construction. Note that indices on the faces on the right are considered modulo $d$.}
 \label{fig:degree_construction}
 \end{figure}
\end{proof}

\begin{lem}\label{thm:lem79_planar_maxspan}
 Let $G=(V,E)$ be a planar graph with some given planar embedding. Let $H=(V,E')$ be the graph, that is obtained from $G$ by removing all edges on the exterior face. Let $T'=(V,F')$ be a maximal spanning forest of $H$. Then there exists a maximal spanning forest $T=(V,F)$ of $G$, such that $er(G,T)\le er(H,T')+2$, and $vr(G,T)\le vr(H,T')+degree(G)$.
\end{lem}
\begin{proof}
 Let $K$ be the graph with vertices in $G$ and edges in $T'$, or in $G$, but not in $H$, i.e.\ $K=(V,(E-E')\cup F')$. Also let $T=(V,F)$ be a maximal spanning forest of $K$, obtained by adding edges from $E-E'$ to $T'$, meaning that $T'\subseteq T$.
 
 Considering fundamental cycles in $K$, relative to $T$, it can be seen that each of these will form the boundary of an interior face in $K$. Also, since every edge is adjacent to at most two interior faces and each vertex adjacent to at most $degree(G)$ interior faces, it follows that $er(K,T)\le2$ and $vr(K,T)\le degree(G)$.
 
 Now, $T$ is a maximal spanning forest of $G$ as well, and every fundamental cycle in $G$, either is a fundamental cycle in $H$, or a fundamental cycle in $K$. This means that $er(G,T)\le er(K,T)+er(H,T')\le er(H,T')+2$, and also that $vr(G,T)\le vr(K,T)+vr(H,T')\le vr(H,T')+degree(G)$.
\end{proof}

\begin{lem}\label{thm:lem80_base_outerplanar_maxspan}
 Let $G=(V,E)$ be an outerplanar graph with $degree(G)\le3$. Then there exists a maximal spanning forest $T=(V,F)$, with $er(G,T)\le2$ and $vr(G,T)\le2$.
\end{lem}
\begin{proof}
 Since all vertices of an outerplanar graph lie on the exterior face, any cycle must consist of at least one edge that lies on the exterior face as well. By removing all edges on the exterior face of the outerplanar graph $G$, a tree or forest $T'=(V,F')$ is thus obtained. Clearly, it is $er(T',T')=vr(T',T')=0$. The claim follows directly as in Lemma \ref{thm:lem79_planar_maxspan}, by observing that each vertex is adjacent to at most 2 interior faces.
\end{proof}

\begin{lem}\label{thm:lem81_step_kouterplanar_maxspan}
 Let $G=(V,E)$ be a $k$-outerplanar graph with $degree(G)\le3$. Then there exists a maximal spanning forest $T=(V,F)$ with $er(G,T)\le2k$, and $vr(G,T)\le3k-1$.
\end{lem}
\begin{proof}
 The claim is shown by induction to $k$, the outerplanarity index. The base case where $k=1$ was shown in Lemma \ref{thm:lem80_base_outerplanar_maxspan}. For $k\ge2$, by removing all edges on the exterior face of $G$, the degree of any vertex on the outer face is equal to $0$ or $1$, and thus, the remaining graph is $(k-1)$-outerplanar. Lemma \ref{thm:lem79_planar_maxspan} and application of inductive reasoning yield the result.
\end{proof}

\begin{thm}\label{thm:thm71_tw_maxspan}
 Let $T=(V,F)$ be a maximal spanning forest of graph $G=(V,E)$. The treewidth of $G$ is at most $\max\{vr(G,T),er(G,T)+1\}$.
\end{thm}
\begin{proof}
 Let $T'$ be the tree obtained by adding an extra vertex in the middle of each edge in $T$, i.e.\ $T'=(V\cup F,F')$, where $F'=\{(v,e)|v\in V,e\in F,\exists w\in V:e=(v,w)\}$. The construction of sets $X_i$, $i\in V\cup F$ is shown, so that $(\{X_i|i\in V\cup F\},T'=(V\cup F,F'))$ is a tree decomposition of $G$.
 
 First, every vertex $v\in V$ is added to $X_v$ and for every edge $(v,w)\in F$, both endpoints are added to $X_{(v,w)}$. Second, for every edge $(v,w)\in E-F$, one endpoint is chosen arbitrarily, say $v$. This vertex is added to each $X_x$, for all vertices $x\in V$, $x\not=w$ that appear on the fundamental cycle of $(v,w)$, but not in $X_w$. The same vertex $v$ is also added to $X_e$, for all edges $e\in F$ that appear on the fundamental cycle of $(v,w)$ as well.
 
 It is easily verified that in this way, a tree decomposition of $G$ is obtained. Furthermore, for all $v\in V$, it is $|X_v|\le1+vr(G,T)$, while for all $e\in F$, it is $|X_e|\le2+er(G,T)$. Thus, the width of this tree decomposition is at most $\max\{vr(G,T),er(G,T)+1\}$.
\end{proof}

\begin{lem}\label{thm:lem16_tw_minor}
 If $G$ is a minor of $H$, then $treewidth(G)\le treewidth(H)$.
\end{lem}
\begin{proof}
 In the case where $G$ is a subgraph of $H$, then removing any occurrences of vertices not in $G$ from the tree decomposition of $H$, yields a tree decomposition of $G$ of equal, or smaller treewidth. That is, if $G=(V,E)$ is a subgraph of $H=(V',E')$ and $(\{X_i|i\in I\},T=(I,F))$ is a tree decomposition of $H$, then $(\{X_i\cap V)|i\in I\},T=(I,F))$ is a tree decomposition of $G$ that satisfies the claim.
 
 In the case where $G$ is obtained by contracting an edge $(v,w)$ of $H$ to a vertex $x$, by replacing all occurrences of $v$ and $w$ in the sets $X_i$ of the tree decomposition of $H$, by occurrences of $x$, one gets a tree decomposition of $G$ of equal, or smaller treewidth. That is, if $(\{X_i|i\in I\},T=(I,F))$ is a tree decomposition of $H$ with treewidth $k$, then $(\{X_i'|i\in I\}, T=(I,F))$ is a tree decomposition of $G$, with treewidth $\le k$, where $X_i'=X_i$, if $v,w\notin X_i$ and $X_i'=(X_i-\{v,w\})\cup\{x\}$, if $v\in X_i$ or $w\in X_i$.
\end{proof}

\begin{thm}\label{thm:thm83_tw_upbound}
 The treewidth of a $k$-outerplanar graph $G=(V,E)$ is at most $3k-1$.
\end{thm}
\begin{proof}
 In the case where $k=1$, Lemma \ref{thm:lem78_tw2} suffices. Subsequently, the claim is shown for the case where $k\ge2$. By Lemma \ref{thm:lem82_degree}, there exists a $k$-outerplanar graph $H$, such that $G$ is a minor of $H$, and $degree(H)\le3$. By Lemma \ref{thm:lem81_step_kouterplanar_maxspan}, there exists a maximal spanning forest $T$ of $H$, such that $er(H,T)\le2k$ and $vr(H,T)\le3k-1$. By Theorem \ref{thm:thm71_tw_maxspan}, $treewidth(H)\le\max\{3k-1,2k+1\}=3k-1$. By Lemma \ref{thm:lem16_tw_minor}, $treewidth(G)\le3k-1$.
\end{proof}

A similar result was shown by Robertson and Seymour \cite{Robertson198449}, yet based on the notion of \emph{radius}. The radius of a planar graph is the maximum distance of a face to the exterior face, calling two faces adjacent when they share a vertex. The proof is omitted here.

\begin{thm}\label{thm:radius}
 The treewidth of a planar graph with radius $d$ is at most $3d+1$.
\end{thm}

\section{Computing a tree decomposition}
This section introduces the algorithm that computes a tree decomposition of a $k$-outerplanar graph of width at most $3k-1$. An overview of the algorithm is presented first, then each step is further discussed separately and finally, the computational complexity of each step (temporal and spatial) is analyzed. 

\paragraph{Overview of the algorithm.} The algorithm is divided into five distinct steps, each based on ideas arising from results of the previous section. The correctness of each step is based on the corresponding result. The overall complexity of $O(kn)$, in both time and space, is shown in the subsequent analysis for each step. Although each step of the algorithm is directly based on specific results, the underlying ideas cannot be trivially applied to obtain an algorithm of the proposed efficiency. Several algorithmic problems arise  when trying to implement the results of the previous section within the restrictions of linear time and space, the solutions to which are mentioned in the analysis of each step.

The algorithm receives a $k$-outerplanar (embedding of) graph $G=(V,E)$  as input and produces a tree decomposition $\mathcal{T}_G=(\{X_i|i\in I\},T'=(I,F'))$ of $G$ as output, where $\max_{i\in I}|X_i|-1\le3k-1$. The outerplanarity index $k$ can also be considered as input, or alternatively, it can be computed in a separate preprocessing step, obviously in linear time for a given embedding (in general, $k$ can be computed in quadratic time \cite{Kammer:2007:DSK:1778580.1778615}). The main steps of the algorithm are the following.
\begin{enumerate}[i.]
 \item\label{step_I} If $k=1$, construct a tree decomposition $\mathcal{T}_G$ of width at most 2 of $G$, according to Lemma \ref{thm:lem78_tw2}. In the case where $k\ge2$, move to Step \ref{step_II}.
 \item\label{step_II} If there exists some vertex $v\in V$ with $degree(v)>3$, then expand $v$ according to Lemma \ref{thm:lem82_degree}. When no more such vertices exist, move to Step \ref{step_III}.
 \item\label{step_III} Construct a maximal spanning forest $T=(V,F)$ of $G$, with $vr(G,T)\le3k-1$, according to Lemmas \ref{thm:lem81_step_kouterplanar_maxspan}, \ref{thm:lem79_planar_maxspan} and \ref{thm:lem80_base_outerplanar_maxspan}.
 \item\label{step_IV} Construct a tree decomposition $\mathcal{T}_G$ of width at most $3k-1$ of $G$, based on $T$ and according to Theorem \ref{thm:thm71_tw_maxspan}.
 \item\label{step_V} For all vertices $v$ that were expanded during Step \ref{step_II}, modify the tree decomposition $\mathcal{T}_G$ of $G$, according to Lemma \ref{thm:lem16_tw_minor}.
\end{enumerate}

\paragraph{Step \ref{step_I}: Decomposing an outerplanar graph.} In the case where $k=1$, the input graph $G$ is an outerplanar graph. In this case, the treewidth of $G$ is at most 2 and a tree decomposition of this width can be computed as follows. First, a vertex $v$ of degree 2 is chosen and removed from the graph, while an edge is added between the two neighboring vertices of $v$, if needed. Vertices of degree 1 are handled in the same way.\footnote{Handling isolated vertices is considered trivial.} Since $G$ is outerplanar, there is always at least one vertex of degree 1, or 2, and since the resulting graph remains outerplanar every time a vertex is removed, this procedure can be repeated until only one edge remains in the graph.

At this point, the construction of the tree decomposition $\mathcal{T}_G$ begins, by adding a node to $\mathcal{T}_G$ that contains both endpoints of the only remaining edge. Then, for every removed vertex $v$, a node containing $v$ and both its neighbors is added to $\mathcal{T}_G$, adjacent to the previously added node that contains the neighbors of $v$. 
When this sequence of node additions is finished, the tree decomposition $\mathcal{T}_G$ is complete and by construction, has width at most 2.

\paragraph{Step \ref{step_II}: Expanding vertices.} The purpose of this step of the algorithm is the transformation of $G$ into a graph $H$, where the degree of any vertex is no higher than 3, such that $G$ is a minor of $H$. Since ideas used in the following steps require a maximum degree of 3, this step ensures that this is indeed the case. As explained in the proof of Lemma \ref{thm:lem82_degree}, this can be accomplished by first computing the layer of each face and then going through all the vertices of $G$ and expanding those, whose degree $d$ is higher than 3.

This involves replacing the vertex by a path of $d-2$ vertices and the reconnection of its edges to these new vertices, so that a face of lowest layer will be placed adjacent to all others. In general, the faces this vertex belongs to are arranged in the way illustrated by Figure \ref{fig:degree_construction}. Note that the order in which the expansions are applied to the vertices of $G$ makes no difference to the final outcome. Contraction and removal of edges and vertices as mentioned in the definition of a
minor (Definition \ref{def:minor}) can be seen as a reverse action to the expansion of a vertex in this way, a fact used in the final step. After the completion of this step, all vertices have a maximum degree of 3.

\paragraph{Step \ref{step_III}: Finding a useful spanning forest.} In this step a maximal spanning forest $T$ of $G$ is constructed, with $vr(G,T)\le3k-1$ (even if it was transformed in the previous step, the input graph is called $G$). This step of the algorithm operates in two distinct stages, the \emph{stripping} stage and the \emph{building} stage. During the stripping stage, all edges on the outer face of the graph are repeatedly removed and placed in corresponding graphs to be used later. During the building stage, $T$ is constructed by repeatedly adding some edges from the graphs containing the edges removed in the stripping stage.

In more detail, the stripping stage consists of $k$ stripping steps. During each step $i$, the $(k-i)$-outerplanar graph $G_{k-i}$ is obtained by removing all edges on the outer face from the remaining graph $G_{k-i+1}$ of the previous step. The subgraph of removed edges of each step is denoted by $R_{k-i}$. After $k-1$ such steps, an outerplanar graph $G_{1}$ is obtained. Since $G_1$ may not be a forest, an additional step is required. By once again removing all edges on the outer face (placing them in $R_0$), the outerplanar forest $G_0=T_0$ ensues and the stripping stage is concluded.

After the stripping stage, the maximal spanning forest $T$ of the original graph $G$ is iteratively constructed during the building stage. This stage consists of the same number of steps, with a direct correspondence between the steps of the two stages. In particular, the objective of each building step $j$ is to construct a maximal spanning forest $T_j$ for the graph $G_j$, obtained during the corresponding stripping step.\footnote{Indices differ between the two stages, as their direction is reversed.} During each building step $j$, a number of edges is chosen from the graph $R_{j-1}$ and added to $T_{j-1}$, which produces $T_j$. Since the constructs are maximal spanning forests, the edges can be chosen and added arbitrarily, as long as the following hold.
\begin{itemize}
 \item There is no cycle in the resulting graph, i.e.\ $T_j$ is a forest.
 \item The number of edges added is maximal, i.e.\ adding another edge from $R_{j-1}$ will result in a cycle in $T_j$.
 \item Each $T_j$ consists of exactly one tree for every connected component of $G_j$.
\end{itemize}

In order to enforce the above three conditions for the addition of edges to the forest of the previous step, while staying within the computational bounds of linearity, a more detailed analysis is required. The main problem here is keeping track of the connected components of each graph $G_j$ and managing the number of trees in each $T_j$, so a direct correspondence between the two is maintained. A naively direct approach to this task can result in non-linear computational time, or space requirements.

For ease of exposition, the number of connected components of $G_j$ (and so the number of trees in $T_j$) can be assumed to be one for each step, without loss of generality. This assumption is valid, since the only choice for addition of edges is between those included in the graph $R_{j-1}$, but there is no edge in this graph that would connect the different components of $G_j$ (those edges, if any, were removed at a different step). Because of this fact, each component of $G_j$ will be handled independently by the algorithm, even without assuming a single target tree. This assumption enables the discussion of each step to focus on the connection of all trees of $T_{j-1}$ in a maximal way and without creating any cycles. 

Next, the structure of the graph of remaining edges $R_{j-1}$ is examined. By construction, any graph $R_{j-1}$ consists of edges (as well as the vertices they connect) that were removed because they were all the edges on the outer face of $G_j$. This means that any graph $R_{j-1}$ is an outerplanar graph and must consist of a number of simple cycles and paths connected to each other, with every cycle defining an empty face, while paths (not contributing to a simple cycle) are connected together in the form of trees. The trees that make up the forest $T_{j-1}$ were originally enclosed by these cycles and for that reason, the existence of at least one such cycle in $R_{j-1}$ is assumed. A somewhat general example of the potential structure of an $R_{j-1}$ graph is presented in Figure \ref{fig:removed_graphs}.

\begin{figure}[htbp]
 \centerline{\includegraphics[width=130mm]{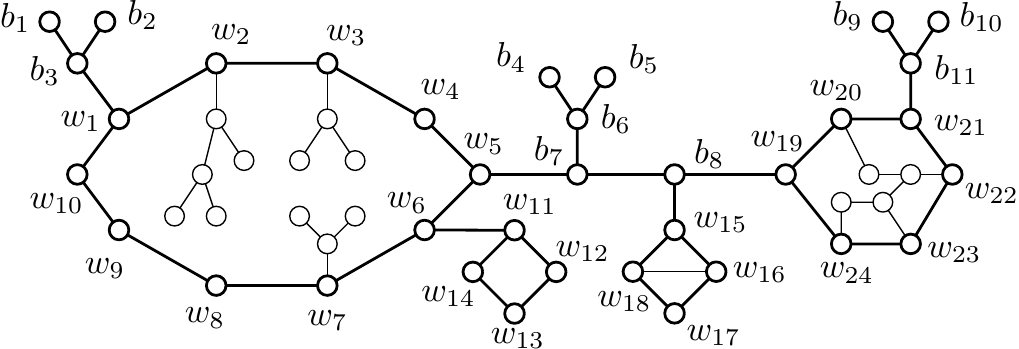}}
 \caption{An example of the potential structure of any $R_{j-1}$ graph (bold) in addition to the forest $T_{j-1}$ (subtle) examined in the same step.}
 \label{fig:removed_graphs}
 \end{figure}
 
Figure \ref{fig:removed_graphs} illustrates the addition of a graph $R_{j-1}$ to the corresponding forest $T_{j-1}$, the parts defining the first appearing in bold. Note the existence of multiple cycles connected by a collection of edges that alone would form distinct trees. In the following, the cycles in $R_{j-1}$ are called \emph{wheels} and the tree-like connected paths are called \emph{branches},\footnote{Not to be confused with notions relating to branchwidth, discussed in the following section.} mostly for clarity. In the above example, all vertices $b_x$ belong to a branch and all vertices $w_y$ belong to a wheel. Without the assumption of only one target tree, several similar constructions would exist elsewhere, with no effect on each other. The goal here is apparent: to choose the maximal set of edges from $R_{j-1}$ to connect $T_{j-1}$ and all vertices that appear in the addition of the two, without leaving any cycles completely intact.

Due to the maximum degree of any vertex being 3, the following fact holds about the connections between wheels and branches in $R_{j-1}$, as well as the trees of $T_{j-1}$. 
\begin{fact}\label{thm:fact_branch_wheel}
 No wheel in any graph $R_{j}$ shares any vertex with another wheel and no tree in a forest $T_{j}$ shares any vertex with a branch in $R_{j}$.
\end{fact}

Since only wheels and branches exist in $R_{j-1}$, the first part of Fact \ref{thm:fact_branch_wheel} means there must always exist a branch between two wheels and furthermore, that this branch is unique for this pair. The second part implies that trees in $T_{j-1}$ can only be connected to a wheel, through some of their leaves, which cannot also belong to a branch. Finally, apart from the path around the wheel, any two vertices on a wheel can only be connected through some tree of $T_{j-1}$, as are vertices $w_{20},w_{22},w_{23}$ and $w_{24}$ in the example of Figure \ref{fig:removed_graphs}.

The above observations allow for a distinction and characterization of all vertices appearing in $R_{j-1}$ in three distinct categories, as follows.
\begin{itemize}
 \item A vertex is characterized as $\alpha$, if it belongs to both a branch and a wheel, connecting the two (examples are $w_1,w_{5},w_{15},w_{19}$ and $w_{21}$ in Figure \ref{fig:removed_graphs}).
 \item A vertex is characterized as $\beta$, if it belongs to both a wheel and a tree of $T_{j-1}$, connecting the two (examples are $w_2,w_{16},w_{18},$ and $w_{24}$).
 \item A vertex is characterized as $\gamma$, if it only belongs to either a wheel, or a branch, meaning all the remaining vertices of $R_{j-1}$ (examples are $b_1,w_4,w_8,$ and $b_{11}$). 
\end{itemize}

After this characterization of the vertices in $R_{j-1}$, the process of combining graphs $T_{j-1}$ and $R_{j-1}$ into a maximal spanning forest $T_{j}$ of $G_j$ can be outlined. First, all edges connecting two $\gamma$ vertices are added to the graph. There is no dispute over the addition of these edges, as some of them are the only connection points between the vertices they connect and the rest of the graph (edges in a branch), while the rest can be safely added since prevention of cycles through them is handled in the following. Second, of the three edges of all $\alpha$ vertices, add the edge that doesn't belong to the wheel, in order to connect the wheel with the rest of the graph. Upon exhaustion of these two cases, only the wheels of $R_{j-1}$ remain (partly) unspoken for.

By the above, the examination of each wheel can be discussed separately.\footnote{This does not imply an algorithmic separation of each wheel, only of analytical presentation.} Begin by selecting a vertex of type $\alpha$, that connects the wheel to a branch and fix an arbitrary directional characterization of the wheel. Moving through vertices according to one direction, add any encountered edges, until a vertex of type $\beta$ is reached. At this point, traverse the entire tree of $T_{j-1}$ that this vertex connects and identify all leaves of the tree marked as $\beta$. These are the connections of the tree to the wheel and should have no other paths connecting them, outside the tree itself.

If there are multiple $\beta$ leaves identified for this tree, continue along the wheel in the same direction, adding all edges between any two of the selected $\beta$ vertices, except one. The edge that is not added can be arbitrary, or chosen to be the last edge before every second $\beta$ vertex. After this is done, continue adding edges around the wheel until another $\beta$ vertex is reached, at which point repeat the above procedure. When the final vertex is reached, do not add the edge connecting it to the first vertex. Applying the above to every wheel in $R_{j-1}$ concludes the construction of $T_{j-1}$.

After all such $T_j$ are constructed in this way, the maximal spanning forest $T$ is complete.

\paragraph{Step \ref{step_IV}: Creating the tree decomposition.} During this step, a tree decomposition $\mathcal{T}_G=(\{X_i|i\in I\}, T'=(I,F'))$ of $G=(V,E)$ is constructed, based on the previously computed $T=(V,F)$ and having width at most $3k-1$. The tree $T'$ of the tree decomposition is obtained by adding an extra vertex in the middle of every edge in $T$, i.e.\ the set of nodes is $I=V\cup F$ and the set of edges is $F'=\{(v,e)|v\in V,e\in F,\exists w\in V:e=(v,w)\}$. The sets $X_i$ that compose the nodes of $\mathcal{T}_G$ are constructed in the following way.
\begin{itemize}
 \item Every vertex $v\in V$ is added to the corresponding set $X_v$.
 \item For every edge $e=(v,w)\in F$ of the spanning forest $T$, both vertices $v$ and $w$ are added to the corresponding set $X_{e}$.
 \item For every edge $(v,w)\in E-F$ that is included in $G$ but is missing in $T$, one endpoint is chosen arbitrarily, say $v$. This vertex is added to every set $X_x$ that corresponds to a vertex $x\in V$, appearing on the fundamental cycle of $(v,w)$ in $T$, except for the set $X_w$. The same vertex $v$ is also added to every set $X_e$ that corresponds to an edge $e\in F$, appearing on the fundamental cycle of $(v,w)$ in $T$.
\end{itemize}

As shown in the proof of Theorem \ref{thm:thm71_tw_maxspan}, the above yield the desired tree decomposition $\mathcal{T}_G$ of width at most $3k-1$. While the first two cases can be easily implemented within the bounds of linear complexity, the third case involves the computation of the fundamental cycle of every missing edge in $T$. The number of such edges is $|E|-|V|+1$ and a trivial approach to the computation of every fundamental cycle may result in a higher order of complexity.

By observing a number of facts about the nature of the edges of $G$ that were not included in $T$, however, all fundamental cycles in $T$ can be efficiently computed. The main idea is that a fundamental cycle in $T$ will consist of edges that were the contours (outlining edges) of a collection of consecutive faces in the original graph $G$. Since $T$ is a forest, all the faces previously defined in $G$ will be `open' in $T$, missing some edges. By adding one of the removed edges back into $T$, a closed face in $T$ will be defined, whose contour will be that of a specific collection of consecutive faces of $G$. Thus, the fundamental cycle can be computed by going through the open faces in $T$ via their missing edges, completing the collection into the closed face and identifying the edges that make up its contour.

Let the \emph{stripping number} $s(f)$ of a face $f$ in $G$ be equal to the number of stripping steps $i$ taken by the algorithm  (Step \ref{step_III}), required for that face to become an outer face in the remaining graph $G_{k-i}$. The stripping number of each face will be used in the computation of the fundamental cycles. Note that the stripping number of a face differs from its layer (as mentioned in the proof of Lemma \ref{thm:lem82_degree}), since for the computation of a stripping number only edges on the outer face are removed, while face layering is computed by removing vertices and all their adjacent edges.

\begin{fact}\label{thm:fact_removed_seps_strip}
 The two adjacent faces in $G$ of every edge $e\in E-F$ (not in $T$) are of different stripping number.
\end{fact}

Intuitively, every missing edge in $T$ was a separator of faces of different stripping number in $G$. The claim can be seen to be correct by noticing that the algorithm's scheme for the addition of edges, employed in the building stage of Step \ref{step_III}, will only consider for non-inclusion those edges that participate in a wheel and thus, are separators of faces of different stripping number in $G$. It does not rely on the scheme chosen however, since for any missing edge $(v,w)$, its adjacent faces could not have been of the same stripping number, as that would imply they became outer faces during the same stripping step in some $G_{k-i}$. If that were the case, no other path could exist between $v$ and $w$ in $G_{k-i}$, except for this edge, which means that the edge itself would have been included in $T_{k-i}$.

\begin{fact}\label{thm:fact_removed_forbidded_seq}
 There can be no missing edges between a sequence of faces $f,f',f''$, with stripping numbers $s(f')=s(f)+1$ and $s(f'')=s(f)$.
\end{fact}

Alternatively, while going through faces via their removed edges, there can be no sequence of faces of alternating stripping numbers. As $f$ and $f''$ have a lower stripping number, edges on their other boundaries (not with $f'$) will be removed first, leaving the only possible way of connecting the vertices on their boundaries with $f'$ with the rest of the graph, to be the boundaries themselves.

\begin{fact}\label{thm:fact_nomorestrip_1}
 No edge can exist in $G$, whose two adjacent faces differ in stripping number by more than one.
\end{fact}

\begin{fact}\label{thm:fact_atmoststrip_1}
 For any pair of adjacent faces in $G$, at most one edge between them can be missing in $T$.
\end{fact}

These two are the last needed facts about the stripping numbers of faces in $G$ and missing edges in $T$. The first claim is trivially true by definition, while the second holds due to the maximality of $T$. 

Facts \ref{thm:fact_removed_seps_strip}, \ref{thm:fact_removed_forbidded_seq}, \ref{thm:fact_nomorestrip_1} and \ref{thm:fact_atmoststrip_1} imply that a tree can be defined based on the hierarchy of stripping numbers of the faces of $G$, with its vertices being the faces of $G$ and edges signifying whether an edge exists that is adjacent to both faces and that is not included in $T$. In particular, the \emph{open-face tree} is constructed as follows. First, take a vertex for each face in $G$. Also take a vertex for the outer face and fix it as a root. Then, for each edge in $G$ but not in $T$, add an edge between the two faces that are adjacent to it in $G$. This construction is indeed a tree, by the facts presented above: Fact \ref{thm:fact_removed_seps_strip} implies that there exists no edge in the open-face tree between vertices of the same level. Fact \ref{thm:fact_removed_forbidded_seq} ensures that a vertex can be connected to at most one vertex of the previous level.
Finally, Facts \ref{thm:fact_nomorestrip_1} and \ref{thm:fact_atmoststrip_1} mean that all edges connect vertices of subsequent levels and that multiple edges do not exist between any two vertices, respectively.

By the construction of the open-face tree, the fundamental cycle of every edge can be computed by the addition of the non-common boundaries of all faces belonging to a subtree, rooted at the lowest face adjacent to the edge in question. The leaves of the open-face tree are the innermost faces of $G$. Starting from the leaves and traversing the tree upwards to the root, adding one endpoint of every missing edge that is encountered to the corresponding set $X_i$ of every vertex, or edge that belong to the faces examined (and also in $T$), will eventually yield the result.  

\paragraph{Step \ref{step_V}: Shrinking expanded vertices.} This step is something of an inverse to Step \ref{step_II}, as it modifies the tree decomposition $\mathcal{T}_G$ computed in Step \ref{step_IV}, based on the expansion of vertices during Step \ref{step_II}. The process is quite simple: for any vertex $v$ that was expanded to a path $v_1,\dots,v_{d-2}$ of $d-2$ vertices, simply replacing all occurrences of any $v_j$ vertex in any set $X_i$ of $\mathcal{T}_G$, by an occurrence of $v$, yields a correct tree decomposition for the original graph, of smaller or equal width.

\paragraph{Complexity of the algorithm.} Since the analysis of each step of the algorithm has been completed, its computational complexity is now discussed. The overall complexity of the algorithm in both time and space is $O(kn)$ and each step will be shown to be confined within this bound.

During Step \ref{step_I} the algorithm goes through the vertices of the graph twice, first to deconstruct the graph and then to recreate it, while also completing the tree decomposition. It is easily verified that each of these procedures requires at most $O(n)$ time, while the tree decomposition itself requires $O(n)$ space.

The computation of the layer of each face during Step \ref{step_II} can be executed in $O(k)$ time and the construction involved in the expansion of a vertex can be implemented in linear time, as at most all edges of the graph will be examined, the number of which is $O(n)$.

Step \ref{step_III} consists of two stages, both operating in $k$ steps. During the stripping stage, all edges on the outer face are removed, which requires at most $O(n)$ time. If the use of an adjacency list is assumed, extended with extra fields signifying the specific subgraphs a vertex partakes in, $O(kn)$ space suffices. During the building stage, several algorithmic ideas are used. First, the computation of the number of connected components that each graph consists of can be  (famously) implemented to require linear time. Also, the distinction between the parts of a subgraph $R_j$, as being branches or wheels, and the subsequent characterization of vertices as $\alpha,\beta$, or $\gamma$ can both be done by simple traversals of the subgraph and therefore require at most $O(n)$ time. Additional space for these characterizations is constant. Finally, the selection of edges based on this characterization requires at most $O(n)$ time as well,
since simple traversals of the branches, wheels and trees of $T_{j-1}$ are sufficient.

The construction of the tree decomposition in Step \ref{step_IV} involves the simple  addition of vertices and (endpoints of) edges to their corresponding sets, which doesn't require more than $O(n)$ time. The computation of the fundamental cycles, as described above, is a bit more complicated. First, the computation of the stripping number of each face needs $O(k)$ time and could have actually been implemented during Step \ref{step_III}. The open-face tree can be constructed in $O(n)$ time and stored in $O(n)$ space, as the number of faces of a planar graph and the number of edges are both linear on the number of vertices. Since all vertices and edges in $G$ belong to a constant number of faces each, and also to at most $O(k)$ fundamental cycles, using the tree to obtain the correct additions to the sets of the tree decomposition requires $O(kn)$ time.
Finally, the size of the computed tree decomposition is $O(kn)$, consisting of $O(n)$ nodes, of size at most $O(k)$ each.

The last step of the algorithm involves the traversal of the tree decomposition and the replacement of any occurrences of added vertices, by the vertex that was expanded for their construction. This procedure does not require more than $O(kn)$ time.

\section{Computing a branch decomposition}
This section discusses matters relating to branch decompositions, starting with a theorem stating that the treewidth and branchwidth of any graph are always within a constant factor \cite{Robertson1991153}. The proof is omitted here.

\begin{thm}\label{thm:tree_branch_width}
 Let $G=(V,E)$ be a graph with treewidth $t$ and branchwidth $b$, $E\not=\emptyset$. Then it is $\max(b,2)\le t+1\le \max(\lfloor\frac{3}{2}\cdot b\rfloor,2)$.
\end{thm}

By this theorem, one would expect similar results, as those presented in the previous sections for tree decompositions of $k$-outerplanar graphs, to hold for branch decompositions as well. This section presents the analysis of an algorithm for the computation of a branch decomposition of a $k$-outerplanar graph, of width at most $2k+1$ (as introduced in \cite{hjortas2005}).

In contrast to the previous section, here the algorithm will be presented first, followed by the justification of its correctness. As mentioned in the introduction, this algorithm also uses some of the same notions introduced in \cite{bodlaender1988planar} and mentioned in the previous sections. In particular, the algorithm shares the same first three steps to obtain the maximal spanning forest of the given graph, but instead uses it to compute a branch decomposition of bounded width.

\paragraph{Overview of the algorithm.} This algorithm is considered to follow Step \ref{step_III} of the algorithm of the previous section and thus, the assumed input is a $k$-outerplanar graph $G=(V,E)$ of maximum degree equal to 3 and a maximal spanning forest $T=(V,F)$ of $G$, with $vr(G,T)\le 3k-1$ and $er(G,T)\le 2k$. The branch decomposition $(T_b=(V_b,F_b),\sigma)$ is computed in the following steps, after initializing $T_b$ to be the same as $T$.

\begin{enumerate}[a.]
 \item\label{Step_a} For all edges $e=(u,v)\in F$, where both $u$ and $v$ have degree at least 2 in $G$, add a new vertex $w$ to $T_b$ and replace $e$ by edges $e_1=(u,w)$ and $e_2(w,v)$. Finally, add a new vertex $x$ and another edge $e_3=(w,x)$. Set $\sigma(e)=x$.
 \item\label{Step_b} For all edges $e=(u,v)\in F$, where $u$ has degree equal to 1 in $G$, set $\sigma(e)=u$.
 \item\label{Step_c} For all edges $e=(u,v)\in E-F$, add a new vertex $x$ and a new edge $e'=(u,x)$ to $T_b$. Set $\sigma(e)=x$.
 \item\label{Step_d} Remove any vertex $v$ of degree equal to 1 from $T_b$, if there exists no edge $e\in E$, such that $\sigma(e)=v$. 
 \item\label{Step_e} Remove any vertex of degree equal to 2 from $T_b$ and add an edge connecting its two neighbors.
\end{enumerate}

It can be easily verified that the output of the above procedure is a branch decomposition $(T_b,\sigma)$ of the given graph $G$. None of the above steps introduces any cycles in $T_b$, while the degree of any vertex after their application is either 1, or 3. Finally, as defined in the description of each step, $\sigma$ is a bijection from $E$ to the set of leaf nodes in $T_b$.

Note that the output of this algorithm is a branch decomposition of bounded width for the input graph $G$, which has a maximum degree equal to 3, since all vertices of higher degree were expanded (as previously described). The modification of this branch decomposition to one of bounded width for the original graph is straightforward and is described below, in the proof of Theorem \ref{thm:branch_minors}.

\paragraph{Complexity of the algorithm.}  As mentioned in the introduction, the above procedure requires linear time. In particular, Steps \ref{Step_a}.\ and \ref{Step_b}.\ perform a constant amount of computation for every edge of $F$, while Step \ref{Step_c}.\ deals with all other edges in $E$. Finally, the last two Steps only remove unnecessary vertices. Since the number of edges is $O(n)$, the procedure will not require more than $O(n)$ time.

\paragraph{Justification of correctness.} The above algorithm is based on similar results as those discussed in the previous sections. In particular, a series of lemmas and two theorems suffice for the purpose of showing the correctness of the algorithm, all results due to \cite{hjortas2005}, except for Theorem \ref{thm:branch_minors} \cite{Robertson1991153}. First, a technical fact used in a following proof.

\begin{lem}\label{thm:lem_hjort_17}
 Given a branch decomposition $(T_b,\sigma)$, constructed by the above algorithm, and three edges $f,g,h\in T$, the following statement holds. If $g$ is on the path from $f$ to $h$ in $T$, then the neighbor of $\sigma(g)$ is on the path from the neighbor of $\sigma(f)$ to the neighbor of $\sigma(h)$ in $T_b$. 
\end{lem}
\begin{proof}
 First, after Steps \ref{Step_a}.\ and \ref{Step_b}.\ the above statement is obviously true. Step \ref{Step_c}.\ refers to edges that are not in $T$ and Step \ref{Step_d}.\ only removes leaves that represent no edge. Finally, Step \ref{Step_e}.\ only removes vertices of degree 2, meaning the above statement remains true throughout the procedure.
\end{proof}

The following lemma analyzes the maximum order of single edges.

\begin{lem}\label{thm:lem_hjort_18}
 Let $(T_b,\sigma)$ be a branch decomposition of a $k$-outerplanar graph $G$ obtained by the above algorithm, given some maximal spanning forest $T$ of $G$. For any edge $e$ in $T$, it is $|\omega(e_i)|\le er_e(G,T)+1$, for $i=1,2$, where $er_e(G,T)$ denotes the edge remember number of the particular edge $e$, instead of the maximum over all edges.
\end{lem}
\begin{proof}
 As an intuitive tool, note that removal of an edge $e_i$ from $T_b$ would partition the tree in two disconnected subtrees. This defines two disjoint sets of edges $r,l$ in $G$, or alternatively two colors. An edge $a$ is assigned to one of these sets (colored by one of the colors), depending on which partition of the tree $T_b$ its corresponding $\sigma(a)$ belongs in. The order of edge $e_i$ is thus the number of vertices in $G$ with edges of both colors.
 
 Now, for $i$ being equal to either 1, or 2, let $f$ and $g$ be edges of $T$ such that $g$ is on the path from $f$ to $e_i$ in $T$. In case no two such edges exist, the claim obviously holds. By Lemma \ref{thm:lem_hjort_17}, $\sigma(g)$ belongs to the same partition of the tree as $\sigma(f)$, which means $f$ and $g$ must have been assigned to the same set (color) in $G$. This means that all edges of $T$ on one side of $e_i$ are assigned to one set in $G$ and all edges of $T$ on the other side are assigned to the other set.
 
 For any edge $e$ in $T$ (to which both the $e_i$ in $T_b$ correspond), there is in $G$ a subgraph $R$ of edges (also appearing in $T$) that are colored $r$ on one side of $e$, and on the other side a subgraph $L$ of edges (appearing in $T$ as well) colored $l$. There is also a set $C_e$ of edges of both colors in between $R$ and $L$. Let $(x,y)$ be such and edge of $C_e$. Then the path from $x$ to $y$ in $T$ must contain $e$, since it is the only edge between $R$ and $L$ that also appears in $T$. This means that this set $C_e$ is exactly the set of edges not in $T$, whose fundamental cycle uses $e$, with the edge $e$ itself also in $C_e$.
 
 The middle set $\omega(e_i)$ is the set of vertices with edges of both colors, meaning edges from $C_e$. Since each such edge has one color, it can only contribute one vertex to $\omega(e_i)$, which gives $|\omega(e_i)|\le er_e(G,T)+1$, for $i=1,2$.
\end{proof}

Subsequently, Lemma \ref{thm:lem_hjort_19} looks into the width of the branch decomposition as a whole.

\begin{lem}\label{thm:lem_hjort_19}
 Let $(T_b,\sigma)$ be a branch decomposition of a $k$-outerplanar graph $G$ obtained by the above algorithm, given some maximal spanning forest $T$ of $G$. Then the width of $(T_b,\sigma)$ is at most $\max\{2,er(G,T)+1\}$.
\end{lem}
\begin{proof}
 By Lemma \ref{thm:lem_hjort_18}, the order of edges of $T_b$ that correspond to an edge also in $T$ is at most $er(G,T)+1$. What remains is the order of edges not appearing in $T$. For all such edges $e=(u,v)$, a new vertex $x$ and edge $e'=(u,x)$ are added to $T_b$ and $\sigma(e)=x$ is set, in Step \ref{Step_c}.\ of the algorithm. The order of edge $e'$ is the number of vertices that have adjacent edges in $G$, the path between which in $T_b$ will use $e'$. But this number is always 2, as the only such vertices are exactly the endpoints of $e$, which implies the result.
 
 Alternatively, using the intuitive visualization of colors described in the proof of Lemma \ref{thm:lem_hjort_18}, edge $e'$ is the only edge of leaf $x$ in $T_b$ and removal of this edge would assign the edge $\sigma^{-1}(x)=e$ one color and a different one to all other edges of $G$. This implies at most two vertices having incident edges of different colors, them being the endpoints of $e$.
\end{proof}

As mentioned above, the proof of the following theorem describes the required modification of the algorithm's computed branch decomposition, to one of bounded width for the original graph, before any vertices were expanded for the maximum degree to equal 3.

\begin{thm}\label{thm:branch_minors}
 If $G$ is a minor of $H$, then $branchwidth(G)\le branchwidth(H)$.
\end{thm}
\begin{proof}
 In the case where $G$ has only one or zero edges, by definition its branchwidth is equal to zero. Suppose $G$ has at least two edges and let $(T_b,\sigma)$ be a branch decomposition of $H$, of lowest width. Let also $S$ be a minimal subtree of $T_b$, such that every edge of $G$ has a representative in $S$, and $T_b'$ be obtained by removing all vertices of degree equal to 2 from $S$ and replacing them with edges connecting their former neighbors. If $\sigma'$ is the restriction of $\sigma$ to the edges of $G$, then $(T_b',\sigma')$ is a branch decomposition of $G$, of lower or equal width.
\end{proof}

Finally, the following theorem summarizes the results presented in this section, while also providing the upper bound to the width of the branch decomposition computed by the above algorithm.

\begin{thm}\label{thm:hjort_25}
 Let $G$ be a $k$-outerplanar graph. Then $G$ has branchwidth at most $2k+1$.
\end{thm}
\begin{proof}
 Without loss of generality, $G$ is assumed to be connected. Lemma \ref{thm:lem82_degree} states there exists a $k$-outerplanar graph $H$ of maximum degree equal to 3, such that $G$ is a minor of $H$. Lemmas \ref{thm:lem80_base_outerplanar_maxspan} and \ref{thm:lem81_step_kouterplanar_maxspan} state there exists a maximal spanning forest $T$ of $H$, such that $er(H,T)\le 2k$. If $(T_b,\sigma)$ is a branch decomposition of $H$, obtained by the above algorithm based on $T$, by Lemma \ref{thm:lem_hjort_19} the width of $(T_b,\sigma)$ is at most $\max\{2,er(H,T)+1\}$. This gives $branchwidth(H)\le2k+1$. By Theorem \ref{thm:branch_minors}, it is $branchwidth(G)\le2k+1$.
\end{proof}

\section{Conclusion}
This paper discusses $k$-outerplanar graphs, focusing on matters relating to their treewidth and branchwidth. Algorithmic ideas first presented in the proof of a theoretical upper bound of at most $3k-1$ on their treewidth \cite{bodlaender1988planar} are used to describe algorithms for the computation of both tree and branch decompositions of bounded width for any given $k$-outerplanar graph. Both algorithms require linear time on the size of the input graph, more precisely $O(kn)$ time for the case of a tree decomposition of width at most $3k-1$, and $O(n)$ time for a branch decomposition of width at most $2k+1$. Since the minimum $k$ for which a planar graph is $k$-outerplanar can be computed in $O(n^2)$ time  \cite{Kammer:2007:DSK:1778580.1778615}, the above can be naturally extended to all planar graphs in general.

One question that remains open is whether a tree decomposition of bounded width can be computed in $O(n)$ time. Even the simple task of reading such a tree decomposition would require $O(kn)$ time, since that is the order of its size. On the other hand, using the alternative representation of a tree decomposition in the form of an elimination ordering (encoding a chordal graph), one can reduce this size to $O(n)$. Despite being aware of a method requiring $O(n)$ time that obtains an alternative elimination ordering of size $3n$ from which the tree decomposition of bounded width can be extracted, this author does not know whether a conventional elimination ordering can also be obtained.

\bibliography{citations}{}
\bibliographystyle{amsplain-nodash}
\nocite{*}

\end{document}